\pdfoutput=1
\RequirePackage[T1]{fontenc}
\RequirePackage{fix-cm}

\documentclass[pdftex,twocolumn]{article}
\usepackage{amsthm}     
\usepackage[fleqn]{amsmath}
\usepackage{newtxtext}
\usepackage[varg]{newtxmath}

\title{Explicit method to make shortened stabilizer EAQECC from stabilizer QECC}
\author{%
Daiki Ueno and Ryutaroh Matsumoto\\
Department of Information and Communications Engineering\\
Tokyo Institute of Technology, 152-8550 Japan}
\date{27 May 2022}

\newtheorem{Thm}{Theorem}      
\newtheorem{Example}{Example}      
\newtheorem{Lem}{Lemma}[Thm]    
\newtheorem{Cor}{Corollary}[Lem] 
\newtheorem{remark}{Remark}    
\newcommand{\QED}{\hspace*{\fill}\rule{1eM}{1eM}}

\begin{document}
\maketitle
\begin{abstract}
  In the previous research by Grassl, Huber and Winter, 
  they proved a theorem which can make 
  entanglement-assisted quantum error-correcting codes
  (EAQECC) from general quantum error-correcting codes (QECC).
  In this paper, we prove that the shortened EAQECC is a stabilizer code 
  if the original EAQECC is a stabilizer code.
\end{abstract}

\section{Introduction}

Quantum error correction is an important tool
for realizing large-scale quantum computers
\cite{chuangnielsen}.
Stabilizer codes
\cite{gottesman96,calderbank97,calderbank98}
is a large class of quantum error-correcting codes
(QECC), which allows efficient encoding of quantum information
into quantum codewords, and relatively efficient decoding.
On the other hand,
the most general class of QECCs \cite{devetak05}
allows neither efficient encoding nor decoding.

The stabilizer code is defined as a common eigenspace
of mutually commuting complex unitary matrices (collectively called as stabilizer),
which is a bit
difficult to handle (for conventional coding theorists working
over finite fields). Calderbank et~al.\ \cite{calderbank97,calderbank98}
showed that most aspects of stabilizer codes can be studied
through linear subspaces over finite fields that are self-orthogonal
with respect to the symplectic inner product \cite{grove02}.
The self-orthogonality corresponds to the commutativity
of unitary matrices, but construction and study of self-orthogonal linear spaces
remain more difficult than those of the conventional linear codes
without orthogonality requirements.

Later Brun et~al.\ \cite{brun06} made a break-through that
enables construction of QECC from any linear spaces over finite fields,
and shorten the code length by utilizing preshared entanglement
between an encoder and a decoder while
keeping the number of information symbols and the error correcting capability.
It is called entanglement-assisted quantum error-correcting codes (EAQECCs).
Brun et al.'s proposal was for qubit (binary case), and
EAQECC was generalized to qudits ($q$-ary case) as well \cite{wilde08,luo17,nadkarni21b}.
Galindo et al.\ \cite{galindo19} provided a description of $q$-ary EAQECCs by
linear codes over finite fields.
Those cited studies concerned with EAQECC constructed from a stabilizer.
Recently, Grassl et al.\ showed the most
general framework of EAQECC \cite{grassl22},
and a method of constructing an EAQECC from any entanglement-unassisted
QECC in Theorem 7 of \cite{grassl22}.
A feature of \cite[Theorem 7]{grassl22} is that it is unclear whether or not
constructed EAQECC is based on stabilizer, even when
the original QECC is stabilizer-based.
There is an advantage of an EAQECC being stabilizer-based,
for example, efficient encoding \cite{nadkarni21a} and
decoding \cite{nadkarni21b} procedures are known
for stabilizer-based EAQECCs.
To this direction,
  Galindo et~al.\ \cite[Proposition 5]{galindo19correction} showed a construction of
  EAQECCs from stabilizer codes, and resulting EAQECCs are also stabilizer-based.
  Our contribution relative to \cite{galindo19correction} is that
  our theorem has a weaker assumption than \cite{galindo19correction}
  (see Remark \ref{remrem}), and it enables us to construct a wider set of
  EAQECCs from the same stabilizer code than \cite{galindo19correction}.

We will show that the resulting EAQECC is stabilizer-based
if the original QECC is stabilizer-based.
Our tools are the conventional puncturing and shortening
of linear codes \cite{pless98},
tailored for self-orthogonal linear spaces with respect to the
symplectic inner product.

This letter is organized as follows: 
First, in Section 2, we explain the necessary assumptions and definitions for the proof, 
and then in Section 3, we prove some lemmas necessary for main theorem and show its proof.
Concluding remarks are given in Section 4.

\section{Preparation}
$F_q$ denotes a finite field of order $q$.
For two vectors $\vec{a}$,$\vec{b}\in F^{n}_q$, the Euclidean inner product is defined by
\[\langle \vec{a},\vec{b} \rangle _E=a_1b_1+\dotsb +a_nb_n .\]
($\vec{a}|\vec{b}$)$\in F^{2n}_q$ denotes the vector concatenating $\vec{a}$ and $\vec{b}\in F^{n}_q$.
For two vectors ($\vec{a}|\vec{b}$),($\vec{c}|\vec{d}$)$\in F^{2n}_q$, the symplectic inner product is defined by
\[\langle (\vec{a}|\vec{b}),(\vec{c}|\vec{d}) \rangle _s = \langle \vec{a},\vec{d} \rangle _E - \langle \vec{b},\vec{c} \rangle _E .\]
For a linear code $V\subset F^{2n}_q$, its symplectic dual code is defined by
$V^{\perp s}=\{ (\vec{a}|\vec{b})\in F^{2n}_q \mid$ for all $(\vec{c}|\vec{d})\in V$,  
$\langle (\vec{a}|\vec{b}),(\vec{c}|\vec{d}) \rangle _s = 0 \}$.
In addition, for ($\vec{a}|\vec{b}$)$\in F^{2n}_q$, the symplectic weight is defined by
\[w(\vec{a}|\vec{b})=\sharp \{i\mid(a_i,b_i)\neq (0,0)\} ,\]
where $\sharp$ denotes the number of elements in a set.

$Puncturing$ in this paper refers to making a new linear code $C'\subset F^{2n-2}_q$ from a linear code 
$C\subset F^{2n}_q$ by eliminating the $i$th and the ($n+i$)th components ($1 \leq i \leq n$) of all vectors in $C$.
$Shortening$ in this paper refers to making a new linear code $C'\subset F^{2n-2}_q$ from a linear code 
$C\subset F^{2n}_q$ by selecting vectors in $C$ where the $i$th and the ($n+i$)th components ($1 \leq i \leq n$) are both zero
and then eliminating the $i$th and the ($n+i$)th components of the selected vectors.

An $[n,k,d;c]_q$ stabilizer EAQECC is a linear code $C\subset F^{2n}_q$ with $c=(\dim C-\dim C\cap C^{\perp s})/2$, $k=c+n-\dim C$ and 
$d=\min \{w(\vec{x})\mid\vec{x}\in C^{\perp s}\setminus C\}$.
An $[n,k,d;0]_q$ stabilizer QECC is a linear code $C\subset F^{2n}_q$ with $0=(\dim C-\dim C\cap C^{\perp s})/2$, $k=n-\dim C$ and 
$d=\min \{w(\vec{x})\mid\vec{x}\in C^{\perp s}\setminus \{\vec{0}\}\}$.

\section{Main Theorem and Proof}
\begin{Thm}\label{Thm1}
%
For a linear code $C\subset F^{2n}_q$ with $C\subset C^{\perp s}$, $k=n-\dim C$, $d=\min \{w(\vec{x})\mid\vec{x}\in C^{\perp s}\setminus \{\vec{0}\} \}$
and for all natural numbers $\ell$ satisfying $1 \le \ell \le d-1$, puncturing and shortening can create 
a new linear code $C^{(\ell)}\subset F^{2(n-\ell)}_q$ with $(\dim C^{(\ell)}-\dim C^{(\ell)}\cap C^{(\ell)\perp s})/2=\ell$, $\ell +(n-\ell)-\dim C^{(\ell)}=k$ 
(therefore $\dim C^{(\ell)}=\dim C$), $d\leq \min \{w(\vec{x})\mid\vec{x}\in C^{(\ell)\perp s}\setminus \{\vec{0}\} \}$.
\end{Thm}
\begin{remark}\label{remrem}
  Galindo et~al.\ \cite[Proposition 5]{galindo19correction}
  assumed $2\ell$ is less than the minimum Hamming weight of $C^{\perp s}$
  regarded as an ordinary linear code of length $2n$,
  while the conclusion was the same as Theorem \ref{Thm1}.
  Galindo et~al.' assumption implies our assumption $\ell < d$,
  and Theorem \ref{Thm1} has wider applicability.
\end{remark}

\begin{Example} \label{Exm1}
  The following set is a basis of a $[5,1,3;0]_{2}$ stabilizer QECC $A$
  \[  \begin{Bmatrix}
        (10010|01100), \\
        (01001|00110), \\
        (10100|00011), \\
        (01010|10001)
      \end{Bmatrix}
  .\]
  Then the following set is a basis of its symplectic dual code $A^{\perp s}$
  \[  \begin{Bmatrix}
        (10010|01100), \\
        (01001|00110), \\
        (10100|00011), \\
        (01010|10001), \\
        (00001|10010), \\
        (00000|11111)
      \end{Bmatrix}
  . \]
  When we puncture $A$ at the 3rd and the 8th components, a basis of the punctured code $A^{(1)}$ is as follows
  \[  \begin{Bmatrix}
        (1010|0100), \\
        (0101|0010), \\
        (1000|0011), \\
        (0110|1001)
      \end{Bmatrix}
  . \]
  When we shorten $A^{\perp s}$ at the 3rd and the 8th components, a basis of the shortened code $A^{(1)\perp s}$ is as follows
  \[  \begin{Bmatrix}
        (1010|1011), \\
        (0101|1101), \\
        (0110|1001), \\
        (0001|1010) 
      \end{Bmatrix}
  . \]
  $A^{(1)}$ is a $[4,1,3;1]_{2}$ stabilizer EAQECC and $(A^{\perp s})^{(1)}$ is a symplectic dual code of $A^{(1)}$.
\end{Example}
In order to prove Theorem \ref{Thm1}, we prove the following four lemmas.
\begin{Lem}\label{Lem1.1}
  Let $C_{(p)}\subset F^{2(n-1)}_{q}$ be a linear code made from a linear code $C\subset F^{2n}_{q}$ 
  with the minimum symplectic weight $w_s\geq 2$, by puncturing $C$ once. Then we have $\dim C=\dim C_{(p)}$.
\end{Lem}
\begin{proof}
  Let $\dim C=k'$.
  Suppose that the number of basis vectors in $C_{(p)}$ is less than $C$.
  
  Let $\{\vec{e}_{1},\vec{e}_{2},\dots,\vec{e}_{k'}\}$ be a basis of $C$, and 
  $\{\vec{e'}_{1},\vec{e'}_{2},\dots,$\\$\vec{e'}_{k'}\}$ be vectors made by puncturing it once.
  When the number of basis vectors in $C_{(p)}$ is less than $C$, there are two cases, namely (1) a zero vector exists in $\{\vec{e'}_{1},\vec{e'}_{2},\dots,\vec{e'}_{k'}\}$ 
  and (2) $\{\vec{e'}_{1},\vec{e'}_{2},\dots,\vec{e'}_{k'}\}$ is linearly dependent.

  Case (1) happens when puncturing of $C$ eliminates a non-zero component of a vector with $w_s =1$ in a basis of $C$, 
  which contradicts to $w_s\geq 2$.

  Case (2) happens when there exists $\vec{e'}_{m}$ ($1\leq m\leq k'$) with
  \[ \vec{e'}_{m}=a_1\vec{e'}_1+\dotsb +a_{(m-1)}\vec{e'}_{(m-1)}+a_{(m+1)}\vec{e'}_{(m+1)} 
  \] \[ 
  +\dotsb +a_{k'}\vec{e'}_{k'} ,\]
  where $a_i\in F_q$ ($1\leq i\leq k'$).
  $C$ has a vector
  \[\vec{e}_{m}-(a_1\vec{e}_1+\dotsb +a_{(m-1)}\vec{e}_{(m-1)}+a_{(m+1)}\vec{e}_{(m+1)}\]
  \[+ \dotsb +a_{k'}\vec{e}_{k'}) \]
  of symplectic weight 1, so the assumption $w_s\geq 2$ does not hold.
\end{proof}
From Lemma \ref{Lem1.1} we have, 
\begin{Cor}\label{Cor1.1.1}
For a linear code $C\subset F^{2n}_{q}$ with the minimum symplectic weight $w_s\geq 2$ and for an integer $\ell$ ($1\leq \ell\leq w_s-1$),
let $C^{(\ell)}_{(p)}\subset F^{2(n-\ell)}_{q}$ be a linear code made from $C$ 
by puncturing it $\ell$ times, then we have $\dim C^{(\ell)}_{(p)}=\dim C$. \QED
\end{Cor}
\begin{Example}
  According to Example \ref{Exm1},
  minimum symplectic weight $w_s$ of $A$ is $3$. 
  We have $\dim A^{(1)}_{(p)}=\dim A=4$.
\end{Example}
For a linear code $C\subset F^{2n}_{q}$ with the minimum symplectic weight $w_s\geq 2$, let $C^{\perp s}$ 
be its symplectic dual code. 
Let $\dim C^{\perp s}=k'^{\perp s}$, $\{\vec{e}_{1},\vec{e}_{2},\dots,\vec{e}_{k'^{\perp s}}\}$ be 
a basis of $C^{\perp s}$, where $\vec{e}_{1},\vec{e}_{2},\dots,\vec{e}_{k'^{\perp s}}$ are row vectors.
Define a matrix $M$ by
\[ M= 
\left( \begin{array}{@{\,}c@{\,}} 
  \vec{e}_{1} \\
  \vdots \\ 
  \vec{e}_{k'^{\perp s}} 
\end{array} \right)
= 
\left( \begin{array}{@{\,}ccc@{\,}} 
  e_{1(1)} & \cdots & e_{1(2n)} \\
    \vdots & \ddots & \vdots \\ 
  e_{k'^{\perp s}(1)} & \cdots & e_{k'^{\perp s}(2n)}
\end{array} \right)
. \]
Then we have the following lemma.
\begin{Lem}\label{Lem1.2}
  If a column in the matrix $M$ is the zero vector or 
  there is an index $i(1\leq i \leq n)$ such that the $i$th column in $M$ is a scalar multiple of the ($n+i$)th column in $M$, 
  then the minimum symplectic weight of $C$ is $1$. \QED
\end{Lem}
  By taking the contraposition of Lemma \ref{Lem1.2} we have, 
\begin{Cor}\label{Cor1.2.1}
  If the minimum symplectic weight of $C$ is $2$ or larger, then there is neither column vector in the matrix $M$ whose all components are zero, nor an index
  $i$ $(1\leq i \leq n)$ such that the $i$th column vector in $M$ is a scalar multiple of the ($n+i$)th column vector in $M$. \QED
\end{Cor}
\begin{Lem}\label{Lem1.3}
  For a linear code $C\subset F^{2n}_{q}$ with the minimum symplectic weight $w_s\geq 2$, let $C^{\perp s}$ 
  be its symplectic dual code, then shortening of $C^{\perp s}$ reduces the dimension of $C^{\perp s}$ by $2$.
\end{Lem}
\begin{proof}
  From Corollary \ref{Cor1.2.1} and the minimum symplectic weight of $C$ being $\geq 2$, the matrix $M$ can be transformed by elementary row operations as follows 
  \[ M=
    \left( \begin{array}{@{\,}ccccccc@{\,}} 
      e'_{1(1)} & \cdots & 1 & \cdots & 0 & \cdots & e'_{1(2n)}\\
      e'_{2(1)} & \cdots & 0 & \cdots & 1 & \cdots & e'_{2(2n)}\\ 
      e'_{3(1)} & \cdots & 0 & \cdots & 0 & \cdots & e'_{3(2n)}\\ 
        \vdots  & \      & \vdots & \ & \vdots & \ & \vdots    \\
      e'_{k'^{\perp s}(1)} & \cdots & 0 & \cdots & 0 & \cdots & e'_{k'^{\perp s}(2n)}
    \end{array} \right)
  . \]
  All row vectors after the third row in $M$ is linearly independent, so when shortening the $i$th and the ($n+i$)th columns of $C^{\perp s}$, 
  a new linear code has ($\dim C^{\perp s}-2$) row vectors as its basis.
\end{proof}
\begin{Lem}\label{Lem1.4}
  For a linear code $C\subset F^{2n}_{q}$ with the minimum symplectic weight $w_s\geq 2$, 
  let $C_{(p)}\subset F^{2(n-1)}_{q}$ be a linear code made from $C$ by puncturing $C$ once, 
  and $C^{\perp s}_{(s)}\subset F^{2(n-1)}_{q}$ be a linear code made from $C^{\perp s}$ by shortening $C^{\perp s}$ once,
  $C^{\perp s}_{(s)}$ is the symplectic dual code of $C_{(p)}$.
\end{Lem}
\begin{Example}
  According to Example \ref{Exm1}, 
  $(A^{\perp s})^{(1)}$ is a symplectic dual code of $A^{(1)}$.
\end{Example}
\begin{proof}
  Let $C^{\perp s}\subset F^{2n}_q$ be the symplectic dual code of a linear code $C\subset F^{2n}_q$.
  Let $H$ be the set of vectors in $C^{\perp s}$ whose $i$th and ($n+i$)th components $(1\leq i \leq n)$ of are $0$.
  Then, we have $C \subset H^{\perp s}$. 
  Also, let $G$ be the set of vectors in $C$ whose $i$th and ($n+i$)th components of are changed to $0$.
  Here, we have $G\subset H^{\perp s}$ because the $i$th and the ($n+i$)th components of vectors 
  in $C$ are multiplied by components of vectors in $H$ which are $0$.

  On the other hand, we have $C_{(p)}\subset C^{\perp s}_{(s)}$ because
  $C_{(p)}$ is the set of vectors in $G$ whose $i$th and 
  ($n+i$)th components of are eliminated, and 
  $C^{\perp s}_{(s)}$ is the set of vectors in $H$ which the $i$th and 
  the ($n+i$)th components of are eliminated.

  From Corollary \ref{Cor1.1.1} and Lemma \ref{Lem1.3}, we have
  \[\dim C_{(p)}+\dim C^{\perp s}_{(s)}=\dim C+(\dim C^{\perp s}-2)=2(n-1) .\]
  From the above and $C_{(p)}\subset C^{\perp s}_{(s)}$, $C^{\perp s}_{(s)}$ is the symplectic dual code of $C_{(p)}$.
\end{proof}
From Lemma \ref{Lem1.3} and \ref{Lem1.4} we have,
\begin{Cor}\label{Cor1.4.1}
  For a linear code $C\subset F^{2n}_{q}$ with the minimum symplectic weight $w_s\geq 2$ and $\ell$ ($1\leq \ell\leq w_s-1$),
  let $C^{(\ell)}_{(p)}\subset F^{2(n-\ell)}_{q}$ be a linear code made from $C$ by puncturing it $\ell$ times, and 
  let $(C^{\perp s})^{(\ell)}_{(s)}\subset F^{2(n-\ell)}_{q}$ be a linear code made from $C^{\perp s}$ by shortening it $\ell$ times, 
  $(C^{\perp s})^{(\ell)}_{(s)}$ is the symplectic dual code of $C^{(\ell)}_{(p)}$. $\square$
\end{Cor}
From the above lemmas, we can now prove Theorem 1.
\begin{proof}[Proof of Theorem 1]
  First, from Corollary \ref{Cor1.1.1}, we have 
  \[ \dim C^{(\ell)}_{(p)}=\dim C .\]

  According to Section 4.1 of the reference \cite{galindo19},
  For a linear code $C\subset F^{2n}_q$ with $C\subset C^{\perp s}$,
  \[ C^{(\ell)}_{(p)} \cap (C^{(\ell)}_{(p)})^{\perp s} = C^{(\ell)}_{(s)}\]
  is satisfied. Therefore, 
  \[ \dim C^{(\ell)}_{(p)} \cap (C^{(\ell)}_{(p)})^{\perp s} = \dim C^{(\ell)}_{(s)} = \dim C - 2\ell .\]
  So, from Corollary \ref{Cor1.4.1} we have 
  \[ \ell =c=(\dim C^{(\ell)}_{(p)}-\dim C^{(\ell)}_{(p)}\cap (C^{\perp s})^{(\ell)}_{(s)})/2 \]
  is satisfied.

  Second, we will prove 
  \[d\leq \min \{w(\vec{x})\mid\vec{x}\in (C^{\perp s})^{(\ell)}_{(s)}\setminus \{\vec{0}\}\} .\]
  In other words, we will prove that the minimum symplectic weight of $(C^{\perp s})^{(\ell)}_{(s)}\setminus C^{(\ell)}_{(p)}$ does not decrease from $d$.
  
  The minimum symplectic weight of $(C^{\perp s})^{(\ell)}_{(s)}$ does not decrease from $d$ because $(C^{\perp s})^{(\ell)}_{(s)}$ is created by
  shortening the subset of $(C^{\perp s})$ and shortening is to eliminate components of vectors in $(C^{\perp s})$ which are zero. 
\end{proof}

\section{Concluding remarks}
In this letter, we consider an extra assumption to \cite[Theorem 7]{grassl22}
that an original QECC is stabilizer-based. Our extra assumption
makes resulting EAQECC being also stabilizer-based, which
facilitates efficient encoding and decoding.
It is unclear to the authors
if stabilizer-based EAQECCs can be constructed from a class of QECCs
wider than the stabilizer codes.

\end{document}